\documentclass[journal]{IEEEtran}

\usepackage{cite}

\usepackage{amsmath}
\usepackage{amssymb}
\usepackage{thmtools}
\usepackage{fixmath}
\usepackage{pgfplots}
\usepackage{mathtools}
\usepackage{tikz}
\usepackage{tikzscale}
\usepackage{algorithm}

\usepackage{environ}
\ifCLASSOPTIONdraftcls
\NewEnviron{mueq}{\begin{equation}\BODY\end{equation}}
\newcommand{\tclf}{}
\newcommand{\tcand}{}
\newcommand{\tcnn}{}
\else
\NewEnviron{mueq}{\begin{multline}\BODY\end{multline}}
\newcommand{\tclf}{\\}
\newcommand{\tcand}{&}
\newcommand{\tcnn}{\IEEEnonumber}
\fi

\DeclareMathAlphabet{\mathbit}{OML}{cmr}{bx}{it}
\DeclareMathAlphabet{\mathss}{T1}{cmss}{m}{sl}
\DeclareMathAlphabet{\mathssbold}{T1}{cmss}{bx}{sl}

\DeclareMathOperator{\tr}{tr}

\DeclareMathOperator*{\argmax}{argmax}
\DeclareMathOperator*{\argmin}{argmin}

\usepackage{ucs}
\usepackage[utf8x]{inputenc}

\newcommand{\mb}[1]{\mathbit{#1}}

\newcommand{\zero}{\boldsymbol{0}}
\newcommand{\id}{\mathbf{I}}
\DeclareMathOperator{\diag}{diag}
\DeclareMathOperator{\blockdiag}{blockdiag}

\newcommand{\He}{{\operatorname{H}}}
\newcommand{\Tr}{{\operatorname{T}}}

\newcommand{\st}{\operatorname{s.t.}}

\newcommand{\Expect}[1]{\operatorname{E}\lbrack#1\rbrack}
\newcommand{\ExpectL}[1]{\operatorname{E}\!\left\lbrack#1\right\rbrack}
\newcommand{\MI}[2]{\operatorname{I}\!\left(#1;#2\right)}
\newcommand{\MICond}[3]{\operatorname{I}\!\left(#1;#2\,|\,#3\right)}

\newcommand{\Cvq}{\mb{C}}
\newcommand{\Cvqtilde}{\widetilde{\mb{C}}}
\newcommand{\CvqS}{\Cvq}

\newcommand{\ijInd}{{k,\ell}}

\newcommand{\mnInd}{{m,n}}

\newcommand{\jiIndEq}{k=n~\text{\normalfont and}~\ell=m}

\newcommand{\CvqSijTr}{[\CvqS^\Tr]_\ijInd}
\newcommand{\CvqSji}{\CvqSijTr}

\newcommand{\mt}[1]{\mathrm{#1}}

\newcommand{\nodeS}{\mt{S}}
\newcommand{\nodeRe}{\mt{R}}
\newcommand{\nodeR}{\nodeRe}
\newcommand{\nodeD}{\mt{D}}

\newcommand{\uS}{\mb u}

\newcommand{\vS}{\mb v}
\newcommand{\CvS}{\mb C_{\vS}}
\newcommand{\starCvS}{\CvS^\star}
\newcommand{\wS}{\mb w}
\newcommand{\CwS}{\mb C_{\wS}}
\newcommand{\starCwS}{\CwS^\star}
\newcommand{\zS}{\mb z}

\newcommand{\xS}{\mb x_\nodeS}

\newcommand{\xR}{\mb x_\nodeRe}

\newcommand{\yR}{\mb y_\nodeRe}

\newcommand{\yD}{\mb y_\nodeD}

\newcommand{\nR}{\mb\eta_\nodeRe}
\newcommand{\CnR}{\mb C_{\nR}}
\newcommand{\nD}{\mb\eta_\nodeD}
\newcommand{\CnD}{\mb C_{\nD}}

\newcommand{\HDR}{\mb H_{\nodeD\nodeRe}}
\newcommand{\HDS}{\mb H_{\nodeD\nodeS}}
\newcommand{\HRS}{\mb H_{\nodeRe\nodeS}}

\newcommand{\GDS}{\mb G_{\nodeD\nodeS}}
\newcommand{\GRS}{\mb G_{\nodeRe\nodeS}}
\newcommand{\barGDS}{\mb{\bar G}_{\nodeD\nodeS}}

\newcommand{\Gk}{\GDS}
\newcommand{\Gj}{\GRS}
\newcommand{\barGk}{\barGDS}

\newcommand{\CzSxR}{\mb{R}}
\newcommand{\CzSxRtilde}{\widetilde{\mb{R}}}
\newcommand{\CzSxRLong}{\mb{C}_{\left[ \substack{\zS \\ \xR} \right]}}

\newcommand{\Amat}{\mb A}

\newcommand{\DS}{\mb D_\nodeS}
\newcommand{\DR}{\mb D_\nodeR}
\newcommand{\PS}{P_\nodeS}
\newcommand{\PR}{P_\nodeR}
\newcommand{\dRS}{d_{\nodeR\nodeS}}
\newcommand{\dDS}{d_{\nodeD\nodeS}}
\newcommand{\dDR}{d_{\nodeD\nodeR}}

\newcommand{\NS}{{N_\nodeS}}
\newcommand{\NR}{{N_\nodeRe}}
\newcommand{\ND}{{N_\nodeD}}

\newcommand{\idNS}{\id}
\newcommand{\idNR}{\id}
\newcommand{\idND}{\id}
\newcommand{\idJoint}{\id}
\newcommand{\idNk}[1][k]{\id}

\newcommand{\RA}{R_\mt{A}}
\newcommand{\RB}{R_\mt{B}}
\newcommand{\hatRA}{\hat{R}_\mt{A}}
\newcommand{\hatRB}{\hat{R}_\mt{B}}

\newcommand{\gev}{\mb F}
\newcommand{\gevf}{\mb f}
\newcommand{\gevL}{\mb \Lambda}
\newcommand{\gevl}{\lambda}
\newcommand{\gevt}{\theta}
\newcommand{\gevN}{N}
\newcommand{\gevLbar}{\mb{\bar \Lambda}}
\newcommand{\gevD}{\mb \Xi}
\newcommand{\gevA}{\mb \Phi}
\newcommand{\gevB}{\mb \Psi}
\newcommand{\Ms}{n}
\newcommand{\ms}{i}

\newcommand{\cpeps}{\epsilon_\text{CP}}
\newcommand{\cpDelta}{\Delta}
\newcommand{\cpL}{L}
\newcommand{\cpU}{U}
\newcommand{\varepsilonprime}{{\varepsilon'}}
\newcommand{\Lam}{\mb\Theta}
\newcommand{\lam}{\theta}

\newtheorem{theorem}{Theorem}
\newtheorem{prop}{Proposition}

\newtheorem{remark}{Remark}

\begin{document}
\title{Maximizing the Partial Decode-and-Forward Rate in the Gaussian MIMO Relay Channel}

\author{Christoph~Hellings,~\IEEEmembership{Member,~IEEE}, Patrick Gest, Thomas Wiegart,~\IEEEmembership{Student~Member,~IEEE},
        and~Wolfgang~Utschick,~\IEEEmembership{Senior~Member,~IEEE}%
\thanks{The authors conducted this research at Technische Universit\"at M\"unchen, Professur f\"ur Methoden der Signalverarbeitung, 80290 M\"unchen, Germany, 
Telephone: +49 89 289-28520, e-mail: hellings@tum.de, patrick.gest@tum.de, thomas.wiegart@tum.de, utschick@tum.de.
C. Hellings is now with Department of Physics, ETH Zurich, 8093 Zurich, Switzerland. T. Wiegart is now with Technische Universit\"at M\"unchen, Institute for Communications Engineering.

The general approach of the proposed algorithm was presented at the 20th International ITG Workshop on Smart Antennas (WSA 2016) \cite{WiHeUt16},
and the (quasi) closed-form solution for the inner problem was presented at the 
19th IEEE International Workshop on Signal Processing Advances in Wireless Communications (SPAWC 2018) \cite{HeGeUt18}.
The explicit calculation of the derivative in Section~\ref{sec:grad} and the modified algorithm for verifying global optimality in Section~\ref{sec:algo} are novel contributions of this journal version.}}

\maketitle

\begin{abstract}
It is known that circularly symmetric Gaussian signals are the optimal input signals for the partial decode-and-forward (PDF) coding scheme in the Gaussian multiple-input multiple-output (MIMO) relay channel,
but there is currently no method to find the optimal covariance matrices nor to compute the optimal achievable PDF rate
since the optimization is a non-convex problem in its original formulation.
In this paper, we show that it is possible to find a convex reformulation of the problem by means of an approximation and a primal decomposition.
We derive an explicit solution for the inner problems as well as an explicit gradient for the outer problem,
so that the efficient cutting-plane method can be applied for solving the outer problem.
As the accuracy of this provably convergent algorithm might be impaired by the previous approximation,
we additionally propose a modified algorithm, for which we cannot give a converge guarantee,
but which provides rigorous upper and lower bounds to the optimum of the original rate maximization.
In numerical simulations, these bounds become tight in all considered instances of the problem,
showing that the proposed method manages to find the global optimum in all these instances.
\end{abstract}

\begin{IEEEkeywords}
\ifCLASSOPTIONpeerreview\vspace*{-3mm}\fi
Gaussian relay channel, multiple-input multiple-output (MIMO), partial decode-and-forward, convex optimization, sensitivity analysis, cutting plane method, generalized eigenvalue decomposition.%
\end{IEEEkeywords}
\IEEEpeerreviewmaketitle

\section{Introduction}
The use of relay stations is a candidate technology for extending the coverage of base stations and increasing the achievable data rates of cell-edge users.
An information theoretical model for such a relaying scenario was first introduced in \cite{Me71}.
This relay channel model describes a three-node network with only a single source, a single relay, and a single destination,
and can be considered as a simplified building block for larger relaying networks,
or as a means for studying basic properties of relaying before turning the attention to more complicated scenarios.
However, even for this basic relaying scenario, the channel capacity is still unknown and the optimal relaying strategy remains an open problem.

To obtain lower bounds to the capacity, researchers have studied
the achievable rates of various relaying strategies (see also \cite[Ch.~9]{Kr07}, \cite[Ch.~16]{GaKi11}) such as amplify-and-forward \cite{LaTsWo04,KrGaGu05,GaMoZa06,FaTh07,RoGa09},
compress-and-forward \cite{CoGa79,KrGaGu05,NgFo11,SiMuViCo10}, 
decode-and-forward (DF) \cite{CoGa79,WaZhHo05,KrGaGu05,SiMuViCo09,NgFo11,GeUt11},
and partial decode-and-forward (PDF)
\cite[Section~16.6]{GaKi11},\cite[Section 9.4.1]{Kr07},\cite{GaAr82,KrGaGu05,LoViHe08,HeGeWeUt14,GeHeWeUt15,JiKi17}.
On the other hand, the so-called cut-set bound (CSB) can be used to obtain an upper bound to the capacity \cite{CoGa79}.

The Gaussian MIMO relay channel, i.e., a relay channel with Gaussian noise and multiple antennas at all terminals 
was considered, e.g., in \cite{WaZhHo05,FaTh07,LoViHe08,RoGa09,SiMuViCo09,SiMuViCo10,NgFo11,GeUt11,GeWeUt13,WeGeUt13,GeWeRiUt14,GeHeWeUt15,JiKi17}.
For such a system,
the PDF scheme is a particularly promising candidate since it was shown in \cite{JiKi17}
that this coding scheme achieves within a constant gap (depending only on the number of antennas at the source and the relay) of the unknown capacity.

The PDF scheme is an extension of DF and belongs to a wider class of relaying strategies proposed in \cite{CoGa79},
where the relay decodes only a part of the source message and applies compress-and-forward to the other part.
In particular, the second part is simply ignored (i.e., it is compressed to a constant value zero) in the PDF scheme \cite{GaAr82}.
This overcomes the issue that a weak source-relay link can become the limiting factor when using DF \cite{KrGaGu05},\cite[Section 9.4.1]{Kr07}.
However, the second part of the source signal, which acts as an interfering signal at the relay, makes the analysis and optimization much more involved.

Just like for DF, the achievable PDF rate is maximized by jointly circularly symmetric Gaussian inputs \cite{GeHeWeUt15},
but the optimal covariance matrices are still unknown.
In case of DF, these covariance matrices can be obtained by solving a convex program \cite{NgFo11,GeUt11},
but the additional interference term makes the rate maximization problem for PDF nonconvex.
Therefore, previous approaches to maximize the PDF rate based on successive convex approximation \cite{WeGeUt13} or zero-forcing \cite{GeWeUt13}
provide only suboptimal solutions.

In this work, we settle this problem in the following manner.
We adopt the primal decomposition approach from \cite{HeGeWeUt14}, where an innovation signal is introduced as an auxiliary signal.
As in the previous conference version \cite{WiHeUt16}, we then approximate the positive-semidefiniteness constraint of the innovation covariance matrix
by a stricter constraint that all eigenvalues have to be greater than or equal to a small positive constant $\varepsilon$,
i.e., the matrix becomes strictly positive-definite.
We then show that the outer problem of this approximate formulation is convex
and can be solved by the efficient cutting plane algorithm.
For the inner problem, we use a generalized eigenvalue decomposition to obtain a (quasi) closed-form solution,
whose derivation is based on results that were obtained in \cite{BuLiPoSh09} and \cite{HeGeUt18}
for the MIMO wiretap channel and the MIMO broadcast channel, respectively.

To overcome the issue from \cite{WiHeUt16} that the approximation using $\varepsilon$ does not allow a rigorous statement about the distance of the obtained solutions
from the true global optimum,
we propose a modified version of the cutting plane method.
For this modified algorithm, we cannot prove that it always converges,
but if it converges, it finds the globally optimal solution up to a desired precision.
Despite this lack of a theoretical convergence guarantee, we observed convergence in all scenarios we considered as numerical examples.
Thus, at least for these scenarios, we settle the conjecture from \cite{WiHeUt16} that the obtained solutions are indeed close-to-optimal.

After introducing the system model and the problem formulation in Section~\ref{sec:model},
we give the details of the proposed reformulation and approximation in Section~\ref{sec:outer}.
Therein, we also show that the approximation leads to a convex outer problem.
In Section~\ref{sec:inner}, we derive the (quasi) closed-form solution to the inner problem.
Based on this, an explicit expression for the gradient of the objective function of the outer problem is derived in Section~\ref{sec:grad}.
Using these ingredients, we summarize the algorithmic solution as well as the modified cutting plane method in Section~\ref{sec:algo}.
The numerical examples in Section~\ref{sec:results} include a comparison to an existing suboptimal approach.

\emph{Notation:}
We use $\id$ for the identity matrix of appropriate size, and $\zero$ for the zero matrix.
The order relations $\succeq$ and $\succ$ have to be understood in the sense of positive-semidefiniteness and positive-definiteness, respectively.
We write $\Expect{\bullet}$ for the expectation, and $\mb C_\mb x$ is the covariance matrix of a random vector $\mb x$.
We use (conditional) mutual information expressions of the form $\MI{\mb x}{\mb y}$ and $\MICond{\mb x}{\mb y}{\mb z}$.
For matrices $\mb A$ and $\mb B$,
we use the Frobenius inner product $\langle\mb A,\mb B\rangle=\tr[\mb B^\He\mb A]$ (e.g., \cite[Sec.~5.2]{HoJo13})
and the Frobenius norm $\|\mb A\|_\mt{F} = \sqrt{\langle\mb A,\mb A\rangle}$, where $\tr[\bullet]$ is the trace of a matrix.
The notation $[\mb A]_{i,j}$ is used for the element in the $i$th row and $j$th column of $\mb A$,
and $\mb A^+$ is the Moore-Penrose pseudoinverse of $\mb A$.

\section{System Model and Problem Formulation}
\label{sec:model}
We consider a Gaussian MIMO relay channel with $\NS$ antennas at the source $\nodeS$, $\NR$ antennas at the relay $\nodeRe$, and $\ND$ antennas at the destination $\nodeD$.
The transmission is described by
\begin{align}
\yR &= \HRS \xS + \nR \\
\yD &= \HDS \xS + \HDR \xR + \nD
\end{align}
where $\mb{H}_\mt{BA}\in\mathbb{C}^{N_\mt{B}\times N_\mt{A}}$ is the channel matrix from node $\mt{A}$ to node $\mt{B}$ with $\mt{A},\mt{B}\in\{\nodeS,\nodeR,\nodeD\}$.
The noise vectors $\nR$ and $\nD$ are assumed to be circularly symmetric Gaussian with zero mean and, without loss of generality, identity covariance matrices
$\CnR=\idNR$ and $\CnD=\idND$.
We assume perfect channel state information and full-duplex transmission with perfect self-interference cancellation at the relay.
An extension to other less idealized scenarios could be considered in future research.

In the PDF scheme (e.g., \cite[Section 16.6]{GaKi11},\cite[Section 9.4.1]{Kr07}) that we consider,
the transmit signal $\xS$ of the source is created as a superposition of two independent parts $\uS$ and $\vS$, where $\uS$ denotes the part that is sent in cooperation with the relay just like in the case of DF.
The second part $\vS$, which is specific to PDF, is directly transmitted without the help of the relay and causes interference at the relay.
We follow the approach from \cite{HeGeWeUt14}, where $\xS$ is further decomposed as
\begin{equation}
\label{eq:decompose_xS}
	\xS = \uS + \vS = \Amat \xR + \wS + \vS = \zS + \wS + \vS.
\end{equation}
where $\wS$ is independent of the relay transmit signal $\xR$ and $\zS$ is fully correlated with $\xR$.

The actual transmission is carried out using a block-Markov coding scheme \cite[Section 9.4.1]{Kr07}, in which causality has to be respected.
Consequently, $\xR$ and $\zS$ must be fully determined by data that has previously been received by the relay, i.e.,
$\zS$ cannot be used to provide any new information.
Instead, $\zS$ is used to serve the user jointly with the relay in a coherent manner.
On the other hand, the signal parts $\wS$ and $\vS$ contain new information that is provided to the relay (to allow for further coherent transmissions in the future) and
to the destination.
We therefore call $\wS+\vS$ the \emph{innovation} signal,
and we call its covariance matrix ${\Cvq = \CvS + \CwS}$ the \emph{innovation covariance matrix} \cite{HeGeWeUt14}.

Our aim is to maximize the achievable data rate of the PDF protocol
(e.g., \cite[Section 9.4.1]{Kr07})
\begin{mueq}
\label{eq:PDFrate}
R=\min\big\{\underbrace{\MICond{\xS}{\yD}{(\uS, \xR)} + \MICond{\uS}{\yR}{\xR}}_{\RA}, \tclf
\,\underbrace{\MI{(\xS,\xR)}{\yD}}_{\RB} \big\}
\end{mueq}
under the constraints
\begin{align}
	\ExpectL{ \xS^\He\xS } &\leq\PS, &
	\ExpectL{ \xR^\He\xR } &\leq\PR 
\end{align}
on the transmit powers of the source and of the relay.
We can use the result that jointly circularly symmetric Gaussian signals
are the optimal input signals for the PDF scheme in the Gaussian MIMO relay channel \cite{GeHeWeUt15},
and we can plug in the decomposition \eqref{eq:decompose_xS}.

The optimization problem we consider is then given by
\begin{IEEEeqnarray}{Rr}
	\label{eq:problem}
	\max_{\substack{\CvS \succeq \zero,\CwS \succeq \zero \IEEEnonumber\\
				\CzSxR \succeq \zero}}
	&\min \left\{ \RA(\CvS,\CwS), \RB(\CvS,\CwS,\CzSxR) \right\} \\
	~\st~	&\tr(\CvS + \CwS) + \tr(\DS \CzSxR	\DS^\He) \leq \PS \IEEEnonumber\\
											&\tr(\DR \CzSxR \DR^\He) \leq \PR
\end{IEEEeqnarray}
with
\begin{mueq}
\label{eq:Ra}
\RA(\CvS,\CwS)=\log_2 \det (\idND + \HDS \CvS \HDS^\He) \tclf
   +\log_2  \frac{\det (\idNR + \HRS (\CvS+\CwS) \HRS^\He)}{\det (\idNR + \HRS \CvS \HRS^\He)},
\end{mueq}%
\begin{mueq}
  \label{eq:Rb}
  \RB(\CvS,\CwS,\CzSxR)=\tclf
   \log_2 \det (\idND + \HDS (\CvS+\CwS) \HDS^\He + 
  \mb{H}
  \CzSxR
  \mb{H}^\He
  )
\end{mueq}
where we have used the joint channel matrix
\begin{IEEEeqnarray}{c}
\mb H = \begin{bmatrix}\HDS & \HDR\end{bmatrix}
\label{eq:defH}
\end{IEEEeqnarray}
the joint covariance matrix
\begin{IEEEeqnarray}{c}
	\CzSxR = \CzSxRLong  =
	\ExpectL{
		\begin{bmatrix} \zS \\ \xR \end{bmatrix}
		\begin{bmatrix} \zS \\ \xR \end{bmatrix}^\He
	}
\end{IEEEeqnarray}
and the selection matrices
\begin{IEEEeqnarray}{rtl}
	\DS = \begin{bmatrix}\idNS& \zero\end{bmatrix}&\:\:and\:\:&\DR = \begin{bmatrix}\zero& \idNR\end{bmatrix}.
\end{IEEEeqnarray}

In \eqref{eq:problem}, we have neglected an additional structural constraint on the joint covariance matrix $\CzSxR$ of $\zS$ and $\xR$
which would ensure that $\zS$ and $\xR$ are fully correlated as required by \eqref{eq:decompose_xS}.
To see that this constraint is automatically fulfilled in the optimum, assume that we have
\begin{equation}
\CzSxR =  
\underbrace{\ExpectL{
		\begin{bmatrix} \Amat\xR \\ \xR \end{bmatrix}
		\begin{bmatrix} \Amat\xR \\ \xR \end{bmatrix}^\He
	}}_{\CzSxR'}+\begin{bmatrix}
	\mb B & \zero \\ \zero & \zero
	\end{bmatrix}
\end{equation}
with 
$\mb B\succeq\zero$.
Then, we could replace $\CzSxR$ by $\CzSxR'$ and $\CwS$ by 
$\CwS'=\CwS+\mb B$.
This would leave the left side of the constraints as well as the rate $\RB$ unchanged
(to see this, note the definition of $\mb H$ in \eqref{eq:defH}),
and the rate $\RA$ would either increase or remain unchanged.
Thus, we do not need to incorporate this constraint in \eqref{eq:problem}.

\section{Approximation and Primal Decomposition}
\label{sec:outer}
As part of a proof in \cite{HeGeWeUt14}, a primal decomposition of \eqref{eq:problem}
into outer and inner problems
was performed by introducing the innovation covariance matrix $\Cvq = \CvS + \CwS$ as an auxiliary variable in the outer problem.
In order to exploit this approach for solving \eqref{eq:problem} numerically, we approximate the constraint set in a way that we obtain a convex optimization problem.

For $\varepsilon\geq0$, we introduce the convex set
\begin{mueq}
\label{eq:constr}
	\mathcal{P}_\varepsilon = \big\{ (\Cvq \succeq \varepsilon\idNS ,\CzSxR \succeq \zero) ~\big|~ \tr(\Cvq) + \tr(\DS \CzSxR \DS^\He) \leq \PS 
	~\tclf\text{\normalfont and} \: \tr(\DR \CzSxR \DR^\He) \leq \PR \big\}
\end{mueq}
which becomes equivalent to the convex constraint set of \eqref{eq:problem} if we plug in $\varepsilon=0$.
Otherwise, i.e., if $\varepsilon>0$, the set $\mathcal{P}_\varepsilon\subset\mathcal{P}_0$ is a slightly tightened version of the constraint set of \eqref{eq:problem}.

We consider the approximated optimization problem
\begin{equation}
\label{eq:primal}
	\max_{(\Cvq,\CzSxR) \in \mathcal{P}_\varepsilon}\!\!\!\! \max_{\substack{\CvS \succeq \zero\\ \CwS \succeq \zero \\ \CvS+\CwS\preceq\Cvq }}\!\!\!\!
	\min \left\{ \RA(\CvS,\CwS), \RB(\CvS,\CwS,\CzSxR) \right\}\!
\end{equation}
which is equivalent to \eqref{eq:problem} if $\varepsilon=0$.
To see why this is true, note that the shaping constraint in the inner minimization is active in the optimum \cite{HeGeWeUt14}.
For $\varepsilon>0$, it provides a feasible solution to \eqref{eq:problem},
and thus a lower bound to the global optimum of \eqref{eq:problem}.
For all numerical examples considered in Section~\ref{sec:results}, 
we obtain the guarantee that this lower bound is numerically tight, i.e., that the global optimum of \eqref{eq:problem} is obtained up to a small error tolerance.

As $\RB$ in \eqref{eq:Rb} depends only on the sum $\CvS+\CwS$,
it is no longer a function of $\CvS$ and $\CwS$ if $\Cvq$ is given.
Thus, the inner optimization over $\CvS$ and $\CwS$ is only needed for $\RA$,
and we can rewrite the problem as
\begin{IEEEeqnarray}{c}
\label{eq:decomposition}
			\begin{aligned}[b]
				\max_{(\Cvq,\CzSxR) \in \mathcal{P}_\varepsilon}~
				&\min \left\{
				\RA^\star(\Cvq), \RB(\Cvq, \CzSxR) \right\}
			\end{aligned}
\end{IEEEeqnarray}
with
\begin{equation}
	\label{eq:Rastar}
	\RA^\star(\Cvq) = 
	\max_{\CvS \succeq \zero, \CwS \succeq \zero} \RA(\CvS,\CwS)
	~~\st~~
	\CvS + \CwS \preceq \Cvq \quad
\end{equation}
and
\begin{equation}
\RB(\Cvq, \CzSxR)=
\log_2 \det (\idND + \HDS \Cvq \HDS^\He + 
  \mb{H}
  \CzSxR
  \mb{H}^\He
  ).
\end{equation}
Below, we show that problem~\eqref{eq:decomposition} is convex,
and in Section~\ref{sec:inner}, we show that \eqref{eq:Rastar} can be solved in (quasi) closed form.
 
\begin{remark}
If we in addition moved the optimization over $\CzSxR$ inside the minimum, noting that only $\RB$ needs to be optimized over $\CzSxR$,
we would obtain the primal decomposition proposed in \cite{HeGeWeUt14}.
For the algorithm proposed in Section~\ref{sec:algo}, we find it more convenient to keep the optimization over $\CzSxR$ together with the optimization over $\Cvq$.
\end{remark}

\begin{theorem}
\label{th:cvx}
The modified optimization problem \eqref{eq:decomposition} with $\varepsilon>0$ is a convex program.
\end{theorem}
\begin{IEEEproof}
The constraint set $\mathcal{P}_\varepsilon$ is convex,
and the objective function is a pointwise minimum,
which is concave if both terms inside the minimum are concave.
Since $\log \det (\mb{X})$ is concave in $\mb{X}\succeq\zero$ \cite[Section 3.1.5]{BoVa09}, the rate expression $\RB(\Cvq,\CzSxR)$ is jointly concave in $\Cvq$ and $\CzSxR$.
Thus, it is sufficient to show that $\RA^\star(\Cvq)$ is concave for $\varepsilon>0$, i.e., for $\Cvq \succ \zero$.

Consider the maximization of the sum rate in a $K$-user MIMO broadcast channel with dirty paper coding \cite{ViJiGo03,WeStSh06} under a shaping constraint
\begin{equation}
\max_{\substack{\mb Q_1 \succeq \zero,\dots,\mb Q_K \succeq \zero\\ \sum_k \mb Q_k\preceq\Cvq} } ~\sum_{k=1}^K \log_2 \frac{\det \left( \idNk + \sum_{i=k}^K \mb{H}_{k} \mb Q_i \mb{H}_{k}^\He\right)} {\det \left( \idNk + \sum_{i=k+1}^K \mb{H}_{k} \mb Q_i \mb{H}_{k}^\He\right)}
\label{eq:max-BC-equivalence}
\end{equation}
with channel matrices $\mb{H}_{k}$ and noise covariance matrices $\idNk$.
For given innovation covariance matrix $\Cvq$,
the maximization in \eqref{eq:Rastar} is mathematically equivalent to \eqref{eq:max-BC-equivalence} for $K=2$ users (cf.~\cite{HeGeWeUt14}).
Since \eqref{eq:max-BC-equivalence} was shown to have zero duality gap for $\Cvq\succ\zero$ \cite{DoRiUt16},
we can consider the Lagrangian dual problem of the maximization in \eqref{eq:Rastar}
and express $\RA^\star(\Cvq)$ as
\begin{mueq}
			\RA^\star(\Cvq)
			=\tclf\min_{\mb\Omega \succeq \zero} \,\, \max_{\CvS \succeq \zero, \CwS \succeq \zero}
				\RA(\CvS,\CwS) + \tr\left(\mb\Omega (\Cvq - (\CvS + \CwS))\right).
			\label{eq:Ra-dual}
\end{mueq}
Assuming some $\Cvqtilde \succ \zero$, this can be bounded from above by
\begin{align}
\RA^\star(\Cvq)\leq\max_{\substack{\CvS \succeq \zero\\ \CwS \succeq \zero}}
				\RA(\CvS,\CwS) &+ \tr \left(\mb{\widetilde{\Omega}} \left(\Cvq - \left(\CvS + \CwS\right)\right)\right) \quad
			\label{eq:sensitivity1}\\
			=\max_{\substack{\CvS \succeq \zero\\ \CwS \succeq \zero}}
						\RA(\CvS,\CwS) &+ \tr \left(\mb{\widetilde{\Omega}} \left(\Cvqtilde - \left(\CvS + \CwS\right)\right)\right) 
						\tcnn\tclf\tcand
						+ \tr\left(\widetilde{\mb\Omega}\left(\Cvq-\Cvqtilde\right)\right) \quad
			\label{eq:sensitivity2}
\end{align}
where \eqref{eq:sensitivity1} is valid for all $\widetilde{\mb\Omega} \succeq \zero$, and \eqref{eq:sensitivity2} is obtained by adding and substracting the term $\tr(\mb{\widetilde{\Omega}} \Cvqtilde)$.

Now consider $\RA^\star(\Cvqtilde)$ and let $\widetilde{\mb\Omega}$ be the optimal Lagrangian multiplier in \eqref{eq:Ra-dual} for this case.
Clearly, \eqref{eq:sensitivity1} also holds for this particular choice of $\widetilde{\mb\Omega}$.
Thus, \eqref{eq:sensitivity2} can be expressed as
\begin{IEEEeqnarray}{rCl}
\RA^\star(\Cvq) &\leq& \underbrace{\RA^\star(\Cvqtilde) + \left\langle \widetilde{\mb\Omega} , \Cvq-\Cvqtilde \right\rangle}_{\hatRA^\star(\Cvq; \Cvqtilde)}
\label{eq:Rahat}
\end{IEEEeqnarray}
showing that $\RA^\star(\Cvq)$ can be bounded from above by the linear approximation \phantom{\LARGE!\!\!}$\hatRA^\star(\Cvq; \Cvqtilde)$ around any $\Cvqtilde \succ \zero$.
Thus, $\RA^\star(\Cvq)$ is concave in $\Cvq$ for $\Cvq\succ\zero$, which concludes the proof.
\end{IEEEproof}

\begin{remark}
The above proof is based on the concept of a sensitivity analysis (cf.~\cite[Section 5.6]{BoVa09}).
\end{remark}
\begin{remark}
If $\varepsilon=0$, the matrix $\Cvq\succeq\varepsilon\idNS$ might have eigenvalues that are zero. In this case, the zero-duality-gap property from \cite{DoRiUt16} does not apply, i.e., \eqref{eq:max-BC-equivalence} might have a duality gap. Thus, we cannot extend Theorem~\ref{th:cvx} to $\varepsilon=0$.
For further implications of a rank-deficient $\Cvq$, see Remark~\ref{rem:rankdef:grad} below.
\end{remark}

\section{Solution to the Inner Problem}
\label{sec:inner}
In this section, we derive a (quasi) closed-form solution to the maximization in \eqref{eq:Rastar} with $\Cvq\succ\zero$.
As this problem is mathematically equivalent to a sum rate maximization \eqref{eq:max-BC-equivalence} in the two-user MIMO broadcast channel with a shaping constraint (see the proof of Theorem~\ref{th:cvx}),
we can follow the lines of \cite{HeGeUt18}, where we have derived a solution to the broadcast channel problem.
The proof relies on a channel enhancement argument that was previously used in \cite{BuLiPoSh09}
to derive a (quasi) closed-form expression of the secrecy capacity of the MIMO wiretap channel.
The concept of channel enhancement for the MIMO broadcast channel was originally proposed in \cite{WeStSh06}.

We use $\Gk=\HDS^\He \HDS$ and $\Gj=\HRS^\He\HRS$, so that we can rewrite \eqref{eq:Ra} as
\begin{mueq}
\label{eq:RaG}
\RA(\CvS,\CwS)=\log_2 \det (\idNS + \Gk \CvS ) \tclf
   +\log_2  \frac{\det (\idNS + \Gj (\CvS+\CwS))}{\det (\idNS + \Gj \CvS )}
\end{mueq}%
where we have used $\det(\id+\mb A\mb B)=\det(\id+\mb B\mb A)$.

\begin{theorem}
\label{th:closed}
Let $\gev$ and $\gevL=\diag\{\gevl_i\}$ such that
\begin{align}
\label{eq:GEVDid}
\gev^\He (\idNS + \CvqS^\frac{1}{2} \Gj \CvqS^\frac{1}{2}) \gev &= \idNS \\
\label{eq:GEVDLambda}
\gev^\He (\idNS + \CvqS^\frac{1}{2} \Gk \CvqS^\frac{1}{2}) \gev &= \gevL
\end{align}
where $\CvqS^\frac{1}{2}$ is the positive-semidefinite square root of $\CvqS\succ\zero$,
and let $\gev_1$ contain the columns of $\gev$ that correspond to the indices $\{i\,|\,\gevl_i>1\}$.
Moreover, let 
\begin{equation}
\gevLbar = \diag\{ \bar\gevl_i \} \quad\text{with}\quad \bar\gevl_i=\max\{\gevl_i;1\}.
\end{equation}
Then, the optimal solution to \eqref{eq:Rastar} is given by
\begin{equation}
\label{eq:closed}
\RA^\star(\CvqS)=\log_2  \det(\gevLbar) + \log_2\det(\idNS + \Gj \CvqS)
\end{equation}
and is attained by
\begin{align}
\label{eq:Qk}
\starCvS &= \CvqS^\frac{1}{2} \gev_1 \gev_1^+ \CvqS^\frac{1}{2}  \\
\label{eq:Qj}
\starCwS &= \CvqS - \starCvS.
\end{align}
\end{theorem}

\begin{proof}[Proof of Theorem~\ref{th:closed}]
As the shaping constraint is fulfilled with equality in the optimal solution \cite{HeGeWeUt14}, we directly have \eqref{eq:Qj}.
Since $\gev_1 \gev_1^+$ is a projection matrix, we have $\zero\preceq\gev_1 \gev_1^+\preceq\id$. This implies that $\zero\preceq\starCvS\preceq\CvqS$, so that the solution in \eqref{eq:Qk}--\eqref{eq:Qj}
is a feasible point for \eqref{eq:Rastar}.

\emph{Achievability:}
In Appendix~\ref{app:detLambda}, the following identities are shown:
\begin{align}
\label{eq:detLambda}
\det (\id + \Gk \starCvS) &=  \det(\gev_1^\He \gev_1)^{-1}\det(\gevLbar),
\\
\label{eq:detid}
\det (\id + \Gj \starCvS) &=  \det(\gev_1^\He \gev_1)^{-1}.
\end{align}
The solution in \eqref{eq:Qk}--\eqref{eq:Qj} thus achieves
\begin{align}
\RA &=\log_2  \frac{\det(\id + \Gk \starCvS)\det(\id + \Gj \CvqS)}{ \det(\id + \Gj \starCvS)}
\\ &=\log_2  \left(\det(\gevLbar)\det(\id + \Gj \CvqS)\right)
\end{align}
which shows that \eqref{eq:closed} is achievable.

\emph{Converse:}
Let 
\begin{equation}
\barGk = \CvqS^{-\frac{1}{2}} (\gev^{-\He} \gevLbar \gev^{-1}  - \id) \CvqS^{-\frac{1}{2}}.
\end{equation}
Since $\gevLbar\succeq\gevL$ and
$\Gk = \CvqS^{-\frac{1}{2}} (\gev^{-\He} \gevL \gev^{-1}  - \id) \CvqS^{-\frac{1}{2}}$,
we have $\barGk\succeq\Gk$.
Thus, replacing $\Gk$ by $\barGk$ leads to an \emph{enhanced} MIMO relay channel.
The solution to \eqref{eq:Rastar} in this enhanced scenario is an upper bound to the solution in the original scenario
since $\barGk\succeq\Gk$ implies 
$ \det(\id + \barGk \CvS)  \geq \det(\id + \Gk \CvS) $ for any $\CvS\succeq\zero$.

Moreover, due to $\gevLbar\succeq\id$, we have $\barGk \succeq \Gj$,
so that the enhanced scenario is a \emph{reversely stochastically degraded} MIMO relay channel (e.g., \cite{GeWeRiUt14}),
where it holds that 
\begin{align}
&\MICond{\xS}{\yD}{(\uS, \xR)} + \MICond{\uS}{\yR}{\xR} \\
\stackrel{\!\!\!\text{\cite{GeWeRiUt14}}\!\!\!}{\leq}~ &\MICond{\xS}{\yD}{\xR}
\\=~&
\log_2  \det (\idND + \barGk (\CvS+\CwS))
\\\leq~&
\log_2\det (\idND + \barGk \CvqS)
\label{eq:UB}
\\=~&
\log_2  \left(\det(\gevLbar)\det(\idNS + \Gj \CvqS)\right).
\end{align}
The last equality is due to 
\begin{align}
\det(\gevLbar) &=
\frac{\det(\gevLbar)}{\det\id}
=
\frac
{\det(\gev^\He (\id + \CvqS^\frac{1}{2} \barGk \CvqS^\frac{1}{2}) \gev)} 
{\det(\gev^\He (\id + \CvqS^\frac{1}{2} \Gj \CvqS^\frac{1}{2}) \gev)}
\\&=
\frac
{\det(\id + \CvqS^\frac{1}{2} \barGk \CvqS^\frac{1}{2})} 
{\det(\id + \CvqS^\frac{1}{2} \Gj \CvqS^\frac{1}{2})}
=\frac{\det(\id + \barGk \CvqS)}{\det(\id + \Gj \CvqS)}.
\end{align}
This shows that the solution to \eqref{eq:Rastar} in the original MIMO relay channel is bounded from above by \eqref{eq:closed}.
\end{proof}

Due to Theorem~\ref{th:closed}, we can compute a solution to \eqref{eq:Rastar} in (quasi) closed form
by solving a generalized eigenvalue problem (cf., e.g., \cite{GoLo96}).
In particular, the matrix $\gevL$ contains the generalized eigenvalues of the matrix pencil $(\gevA,\gevB)$
with $\gevA=\id + \CvqS^\frac{1}{2} \Gk \CvqS^\frac{1}{2}$ and
$\gevB=\id + \CvqS^\frac{1}{2} \Gj \CvqS^\frac{1}{2}$,
while $\gev$ contains the corresponding generalized eigenvectors.
This can be seen from
\begin{align}
\label{eq:LambdaFromGEVD1}
\left.\begin{aligned}\gev^\He\gevA\gev &\stackrel{\eqref{eq:GEVDLambda}}{=} \gevL\\
(\gev^\He\gevB\gev)^{-1} &\stackrel{\eqref{eq:GEVDid}}{=} \id
\end{aligned}\right\} &\Rightarrow~
(\gev^\He\gevB\gev)^{-1} \gev^\He\gevA\gev = \gevL\\
&\Leftrightarrow~ \gev^{-1}\gevB^{-1}\gev^{-\He} \gev^\He\gevA\gev = \gevL\\
&\Leftrightarrow~ \gevA\gev = \gevB\gev\gevL. \label{eq:LambdaFromGEVD3}
\end{align}

\begin{remark}
Since the solution to the generalized eigenvalue problem in \eqref{eq:LambdaFromGEVD3} is ambiguous with respect to a scaling of the columns of $\gev$,
the implication ``$\Rightarrow$'' in \eqref{eq:LambdaFromGEVD1} is not an equivalence ``$\Leftrightarrow$''.
However, choosing this scaling such that \eqref{eq:GEVDid}--\eqref{eq:GEVDLambda} hold is possible for any solution to \eqref{eq:LambdaFromGEVD3}.
\end{remark}

\section{Calculation of Gradients}
\label{sec:grad}
In this section, we derive linear approximations by means of gradients of $\RA^\star(\Cvq)$ and $\RB(\Cvq,\CzSxR)$,
which can then be used for an algorithmic solution to the outer problem.

For the rate expression $\RB(\Cvq,\CzSxR)$, a linear approximation around a point $(\Cvqtilde, \CzSxRtilde)$ can be easily obtained as
\begin{IEEEeqnarray}{rCl}
\label{eq:RbhatFirst}
			\RB(\Cvq,\CzSxR) &\leq& \hatRB(\Cvq,\CzSxR; \Cvqtilde, \CzSxRtilde) \\
				&=& \RB(\Cvqtilde,\CzSxRtilde)
				+ \left\langle \left. \tfrac{\partial \RB}{\partial \Cvq^\Tr}\right|_{\Cvqtilde}
					 , \, \Cvq - \Cvqtilde \right\rangle \tcnn\tclf\tcand\tcand
				\hfill + \left\langle \left. \tfrac{\partial \RB}{\partial \CzSxR^\Tr}\right|_{\CzSxRtilde} 
					 , \, \CzSxR - \CzSxRtilde
				\right\rangle
	\label{eq:Rbhat}
\end{IEEEeqnarray}
where the gradient matrices $\frac{\partial \RB}{\partial \Cvq^\Tr}$ and $\frac{\partial \RB}{\partial \CzSxR^\Tr}$ are given by
\begin{IEEEeqnarray}{rCl}
	\frac{\partial \RB}{\partial \Cvq^\Tr} &=& \frac{1}{\ln 2} \HDS^\He
		\mb{D}^{-1}
		\HDS \\
	\frac{\partial \RB}{\partial \CzSxR^\Tr} &=& \frac{1}{\ln 2} \mb{H}^\He
		\mb{D}^{-1}
		\mb{H}
\end{IEEEeqnarray}
with
\begin{IEEEeqnarray}{C}
\mb{D} = \idND + \HDS \Cvq \HDS^\He + \mb{H} \CzSxR \mb{H}^\He . 
\end{IEEEeqnarray}
The inequality in \eqref{eq:RbhatFirst} holds due to concavity of $\RB$ (see proof of Theorem~\ref{th:cvx}).

For $\RA^\star(\Cvq)$, we have already obtained a linear approximation in \eqref{eq:Rahat}.
Therein, $\widetilde{\mb\Omega}$ is a (concave) subgradient, and according to the sensitivity analysis in \eqref{eq:sensitivity1}--\eqref{eq:sensitivity2},
it can be computed numerically by finding the optimal dual variable (Lagrangian multiplier) $\widetilde{\mb\Omega}$ in \eqref{eq:Ra-dual}.
This approach was pursued in \cite{WiHeUt16}, where $\RA^\star(\Cvq)$ was evaluated by a numerical solver that delivers both the optimal primal and dual variables.
However, as we use the (quasi) closed-form solution \eqref{eq:closed} for the sake of a small computational complexity,
we instead calculate $\widetilde{\mb\Omega}$ by deriving the gradient $\widetilde{\mb\Omega}=\frac{\partial\RA^\star(\CvqS)}{\partial\CvqS^\Tr}$ directly based on \eqref{eq:closed}.

When taking the derivative of \eqref{eq:closed} with respect to $\CvqS$, 
the challenging part is the first summand in
\begin{mueq}
\label{eq:deriv_closed}
\widetilde{\mb\Omega}=\frac{\partial\RA^\star(\CvqS)}{\partial\CvqS^\Tr}
=\frac{\partial \log_2 \det(\gevLbar)}{\partial\CvqS^\Tr} \tclf+
\frac{1}{\ln2} \HRS^\He \left(\idNS + \HRS \CvqS \HRS^\He \right)^{-1} \HRS.
\end{mueq}
Assuming that the generalized eigenvalues in Theorem~\ref{th:closed} are sorted in descending order,
the first summand in \eqref{eq:closed} can be written as
\begin{equation}
\label{eq:closed:first}
\log_2\left(\prod_{i=1}^\NS \bar\gevl_i \right)
 = \log_2\left(\prod_{i=1}^\gevN \gevl_i \right)
 = \sum_{i=1}^\gevN \log_2 \gevl_i
\end{equation}
where $\gevN=|\{i\,|\,\gevl_i>1\}|$.
The missing derivative is thus given by
\begin{equation}
\label{eq:deriv_closed:first}
\frac{\partial \log_2 \det(\gevLbar)}{\partial\CvqS^\Tr} = 
\sum_{i=1}^\gevN \frac{1}{\gevl_i\,\ln2} \frac{\partial\gevl_i}{\partial\CvqS^\Tr}
\end{equation}
which we compute element-wise in the following.

To obtain $[\frac{\partial\gevl_i}{\partial\CvqS^\Tr}]_\ijInd=\frac{\partial\gevl_i}{\partial\CvqSijTr}$, we use the fact that the derivative of a generalized eigenvalue with respect to a parameter $\gevt$ is given by \cite{Ab09}
\begin{equation}
\label{eq:deriv_gevl}
\frac{\partial\gevl_i}{\partial \gevt} = \gevf_i^\He \left(\frac{\partial \gevA}{\gevt}-\gevl_i\frac{\partial\gevB}{\gevt}\right)\gevf_i
\end{equation}
where $\gevf_i$ is the generalized eigenvector corresponding to $\gevl_i$.
This result holds under the assumption that $\gevl_i$ has multiplicity one.
Since matrices with eigenvalues with multiplicities larger than one are a set of measure zero within the set of Hermitian matrices (e.g., \cite[Sec.~A.37]{Me04}),
the assumption is sensible if the channel matrices are drawn from a continuous distributions and $\Cvq$ has full rank.

In \eqref{eq:deriv_gevl}, we need the derivatives
\begin{align}
\label{eq:deriv_gevA}
\frac{\partial \gevA}{\partial \gevt} &= 
\frac{\partial \CvqS^\frac{1}{2}}{\partial \gevt} \HDS^\He \HDS \CvqS^\frac{1}{2} +  \CvqS^\frac{1}{2} \HDS^\He \HDS \frac{\partial \CvqS^\frac{1}{2}}{\partial \gevt}, \\
\label{eq:deriv_gevB}
\frac{\partial \gevB}{\partial \gevt} &= 
\frac{\partial \CvqS^\frac{1}{2}}{\partial \gevt} \HRS^\He \HRS \CvqS^\frac{1}{2} +  \CvqS^\frac{1}{2} \HRS^\He \HRS \frac{\partial \CvqS^\frac{1}{2}}{\partial \gevt}.
\end{align}
Setting 
$\gevt=\CvqSijTr$,
we can then obtain the derivative of $\gevl_i$ with respect to the entries of $\CvqS$.
It thus remains to calculate $\frac{\partial \CvqS^\frac{1}{2}}{\partial\CvqSji}$.

To this end, we use $\CvqS= \CvqS^\frac{1}{2}\CvqS^\frac{1}{2}$ to obtain the derivative
\begin{equation}
\label{eq:deriv_sqrtS}
\frac{\partial \CvqS}{\partial\CvqSji}
=
\frac{\partial \CvqS^\frac{1}{2}}{\partial\CvqSji} \CvqS^\frac{1}{2}
+
\CvqS^\frac{1}{2} \frac{\partial \CvqS^\frac{1}{2}}{\partial\CvqSji}.
\end{equation}
Noting that
\begin{align}
\left[\frac{\partial \CvqS}{\partial\CvqSji}\right]_\mnInd = \begin{cases}
1\quad &\text{if~}\jiIndEq, \\
0&\text{otherwise}
\end{cases}
\end{align}
we can calculate the unknown $\frac{\partial \CvqS^\frac{1}{2}}{\partial\CvqSji}$ by solving the linear system of equations \eqref{eq:deriv_sqrtS}.
As this system is a Sylvester equation \cite[Sec.~2.4.4]{HoJo13}, it
has a unique solution if $\CvqS^\frac{1}{2}$ and $-\CvqS^\frac{1}{2}$ have no common eigenvalues \cite[Th.~2.4.4.1]{HoJo13}.
This is fulfilled under the assumption $\CvqS\succ\zero$
since all eigenvalues of $\CvqS^\frac{1}{2}$ are strictly positive in this case.

The gradient $\widetilde{\mb\Omega}=\frac{\partial\RA^\star(\CvqS)}{\partial\CvqS^\Tr}$ can thus be calculated by combining
\eqref{eq:deriv_closed}, \eqref{eq:deriv_closed:first}, \eqref{eq:deriv_gevl}, \eqref{eq:deriv_gevA}, and \eqref{eq:deriv_gevB} with the solution to \eqref{eq:deriv_sqrtS}.

\begin{remark}
\label{rem:rankdef:grad}
In the calculation of the gradient, we need $\Cvq\succ\zero$ for arguing that there are no repeated eigenvalues as well as for solving the above Sylvester equation,
i.e., we again cannot easily extend the result to the case $\varepsilon=0$.
This is in line with the observation that the sensitivity analysis \eqref{eq:sensitivity1}--\eqref{eq:sensitivity2} does not extend to this case.
If $\Cvq$ has eigenvalues equal to zero, the corresponding eigenvectors can be interpreted as forbidden directions in which no signal power can be used.
This means that Slater's constraint qualifications (see, e.g., \cite[Ch.~5]{BaShSh06}, or \cite[Ch.~1]{BeNe13} for the corresponding concept in semidefinite programming)
are not fulfilled
since the constraint set of \eqref{eq:Rastar} has an empty interior in this case.
As a consequence, it might be the case that the KKT conditions (e.g., \cite[Sec.~4.2]{BaShSh06}) are no longer necessary for an optimal solution, i.e.,
it may happen that no combination of optimal primal and optimal dual variables exists, even though we can find an optimal primal solution via \eqref{eq:closed}.
Indeed, we verified numerically in \cite{WiHeUt16} that there are cases 
where the optimum of the original problem can indeed be attained at a point which does not fulfill the KKT conditions,
i.e., where we cannot find an optimal dual variable $\widetilde{\mb\Omega}$ to be used as a subgradient in the linearization \eqref{eq:Rahat}.
Thus, it is not surprising that the alternative approach of obtaining $\widetilde{\mb\Omega}$ via an explicit gradient calculation
fails as well if $\Cvq$ is rank-deficient.
\end{remark}

\section{Solution to the Outer Problem}
\label{sec:algo}
We have shown that the optimization problem \eqref{eq:decomposition} with $\varepsilon>0$ is convex (Theorem~\ref{th:cvx}), and we have calculated the gradients of both terms inside the minimum operator in the objective function (Section~\ref{sec:grad}).
For solving \eqref{eq:decomposition} with $\varepsilon>0$, we can thus choose from the wide range of derivative-based methods for convex programming.
In the following, we propose a solution using the cutting plane method \cite{Ke60,BaShSh06},
which successively refines linear approximations of a concave function.
Afterwards, we discuss how a solution to the original PDF rate maximization \eqref{eq:problem} can be obtained, i.e., for the case $\varepsilon=0$.

\subsection{Provably Convergent Algorithm for $\varepsilon>0$}
\label{sec:algo:cp}
To apply the cutting plane algorithm, we reformulate \eqref{eq:decomposition} as
\begin{IEEEeqnarray}{rCl}
\label{eq:cp1}
				\max_{\substack{(\Cvq,\CzSxR) \in \mathcal{P}_\varepsilon\\\cpU\in\mathbb{R}}} \cpU &\quad\st\quad & \cpU \leq \RA^\star(\Cvq), \: \cpU \leq \RB(\Cvq, \CzSxR).
\end{IEEEeqnarray}
The optimal value of this problem can be bounded from above by replacing $\RA^\star(\Cvq)$ and $\RB(\Cvq, \CzSxR)$ in the constraints of \eqref{eq:cp1} by their linear approximations \eqref{eq:Rahat} and \eqref{eq:Rbhat} around a finite number of points.
To this end, we consider the problem
\begin{align}
\label{eq:cp2}
	\begin{rcases}
		\begin{aligned}
		&\max_{\substack{(\Cvq,\CzSxR) \in \mathcal{P}_\varepsilon\\\cpU\in\mathbb{R}}} \cpU ~\st& \cpU &\leq \hatRA^\star(\Cvq; \Cvqtilde) \\[-0.3cm]
		&& \cpU &\leq \hatRB(\Cvq, \CzSxR; \Cvqtilde, \CzSxRtilde)
		 \end{aligned}
	\end{rcases}
	{\forall (\Cvqtilde, \CzSxRtilde) \in \mathcal{P}^{(\Ms)}}
\end{align}
where $\mathcal{P}^{(\Ms)} \subset \mathcal{P}_\varepsilon$ contains all points at which a linearization has been performed.
This corresponds to an outer approximation of the constraint set (i.e., a relaxed problem) since the linear approximations overestimate the concave functions $\RA^\star$ and $\RB$.
Problem \eqref{eq:cp2} is a semidefinite program and can be solved by standard solvers such as, e.g., SDPT3 \cite{ToTuTo06}.

The cutting plane algorithm is initialized with an initial set $\mathcal{P}^{(1)}$ with a small number of strictly feasible elements,
i.e., points $(\Cvq,\CzSxR)$ from the interior of the constraint set $\mathcal{P}_\varepsilon$.
For instance, we could use $\mathcal{P}^{(1)}=\{(\alpha\idNS,\beta\idJoint)\}$ with $\alpha>\varepsilon$ and $\beta>0$ chosen in a way that both power constraints in \eqref{eq:constr} are fulfilled with strict inequality.
In the $\Ms$th iteration, problem \eqref{eq:cp2} is solved, which yields an upper bound $\cpU^{(\Ms)}$ to the optimal value of \eqref{eq:decomposition}
and a new candidate point $(\Cvq^{(\Ms)},\CzSxR^{(\Ms)})$.
This point is added to the set of points in $\mathcal{P}^{(\Ms)}$ to obtain a new set $\mathcal{P}^{(\Ms+1)}$, and the gradients from Section~\ref{sec:grad} are calculated to obtain the additional linear approximations.
Moreover, by evaluating the objective function of \eqref{eq:decomposition} at the new candidate point, we can get a lower bound 
\begin{equation}
\cpL^{(\Ms)}=\min\{\RA^\star(\Cvq^{(\Ms)}),\RB(\Cvq^{(\Ms)},\CzSxR^{(\Ms)})\}
\end{equation}
to the optimal value.

The current best solution in the $\Ms$th iteration is the candidate point from iteration number $\argmax_{\ms\in\{1,\dots,\Ms\}} \cpL^{(\ms)}$,
and we get the guarantee that this solution lies at most
\begin{equation}
\cpDelta^{(\Ms)} = \cpU^{(\Ms)} - \max_{\ms\in\{1,\dots,\Ms\}} \cpL^{(\ms)}
\end{equation}
away from the global optimum of \eqref{eq:decomposition}.
This error $\cpDelta^{(\Ms)}$ converges towards zero due to the convergence proof of the cutting plane algorithm in \cite{Ke60}, and we can terminate the algorithm if a desired accuracy $\cpeps$ is reached.

The procedure is summarized in Algorithm~\ref{algo:cp}, and the intuition behind the convergence proof is as follows.
If the optimizer $(\Cvq^{(\Ms)},\CzSxR^{(\Ms)})$ of \eqref{eq:cp2} lies inside $\mathcal{P}_\varepsilon$, we have found an optimal solution to \eqref{eq:decomposition}
and obtain $\cpU{(\Ms)}=\cpL{(\Ms)}$.
In any other case, adding a new linear constraint due to a linearization at $(\Cvq^{(\Ms)},\CzSxR^{(\Ms)})$ 
refines the outer approximation.
This can be interpreted in a graphical manner as adding a tangent hyperplane at $(\Cvq^{(\Ms)},\CzSxR^{(\Ms)})$
in order to cut away a halfspace in which the optimal solution cannot lie.
Since the refined outer approximation leads to a decreased upper bound $\cpU^{(\Ms+1)}$,
the values of the upper bound form a decreasing sequence.
It can be shown (see \cite{Ke60}) that
this sequence converges towards the global optimal of \eqref{eq:decomposition}.

\begin{algorithm}[h]
\caption{Cutting Plane Method for~\eqref{eq:decomposition} with~$\varepsilon>0$}
\label{algo:cp}
Given $\Ms=0$, $\varepsilon>0$, and an initial set $\mathcal{P}^{(1)}$:
\begin{enumerate}\setlength{\itemsep}{6pt}
\item $\Ms\gets \Ms+1$
\label{item:cp_first}
\item $(\Cvq^{(\Ms)},\CzSxR^{(\Ms)},\cpU^{(\Ms)})$ $\gets$ optimizer of \eqref{eq:cp2}
\label{item:cp_sdp}
\item $\mathcal{P}^{(\Ms+1)} \gets \mathcal{P}^{(\Ms)} \cup \{(\Cvq^{(\Ms)},\CzSxR^{(\Ms)})\} $ \label{item:cp_newset}
\item Repeat \ref{item:cp_first})--\ref{item:cp_newset}) until
$\cpDelta^{(\Ms)}\leq\cpeps$
\end{enumerate}
\end{algorithm}

It remains to discuss the quality of the obtained solution in terms of the original PDF rate maximization.
As increasing $\varepsilon$ tightens the constraint set of \eqref{eq:decomposition}, the obtained solution is guaranteed to be feasible for \eqref{eq:decomposition}, but it might be suboptimal.
Unfortunately, it is not clear how strong this suboptimality can become.
In other words, it is an open question whether the optimal value of \eqref{eq:decomposition} is a continuous function of $\varepsilon$.
If this was not the case, it could in principle happen that an arbitrarily small $\varepsilon$ still leads to a large error in the optimal solution.
To overcome this issue, we propose a modified algorithm below.

\subsection{Algorithm for $\varepsilon=0$}
\label{sec:algo:mod}
When using the cutting plane algorithm,
the necessity of having $\varepsilon>0$ does not arise from the semidefinite program \eqref{eq:cp2},
but from the gradient calculation using \eqref{eq:deriv_sqrtS}.
We may thus allow $\varepsilon=0$ in \eqref{eq:cp2} if we have some other means to ensure
that $\mathcal{P}^{(\Ms)}\subset \mathcal{P}_\varepsilonprime$ for some $\varepsilonprime>0$,
i.e., to ensure that we need to evaluate the gradient for full-rank $\Cvq$.

Let $(\Cvq^{(\Ms)},\CzSxR^{(\Ms)})$ denote a candidate point obtained using $\varepsilon=0$.
As $\Cvq^{(\Ms)}$ might be rank-deficient, we propose to perform the orthogonal projection
\begin{equation}
\label{eq:proj}
\Cvq^\star=\argmin_{\Cvq\succeq\varepsilonprime \idNS} \| \Cvq - \Cvq^{(\Ms)} \|_\mt{F}^2
\end{equation}
and to add $(\Cvq^\star,\CzSxR^{(\Ms)})$ to the set $\mathcal{P}^{(\Ms+1)}$ instead of adding
$(\Cvq^{(\Ms)},\CzSxR^{(\Ms)})$.

The solution to this problem is given in the following proposition, whose proof is provided in Appendix~\ref{app:proj}.
\begin{prop}
\label{prop:proj}
The solution to \eqref{eq:proj} is given by
\begin{equation}
\Cvq^\star = \mb V \diag\{\max\{\kappa_i,\varepsilonprime\}\} \mb V^\He
\end{equation}
where $\mb V$ and $\kappa_i$ are obtained from the eigenvalue decomposition $\Cvq^{(\Ms)} = \mb V \diag\{\kappa_i\} \mb V^\He$.
\end{prop}

\begin{algorithm}[h]
\caption{Modified Algorithm for~\eqref{eq:decomposition} with~$\varepsilon=0$}
\label{algo:cpmod}
Given $\Ms=0$, $\varepsilonprime>0$ and an initial set $\mathcal{P}^{(1)}$:
\begin{enumerate}\setlength{\itemsep}{6pt}
\item $\Ms\gets \Ms+1$
\label{item:cpmod_first}
\item  $(\Cvq^{(\Ms)},\CzSxR^{(\Ms)},\cpU^{(\Ms)})$ $\gets$ optimizer of \eqref{eq:cp2} with $\varepsilon=0$ 
\label{item:cpmod_sdp}
\item $\Cvq^\star \gets$ projection of $\Cvq^{(\Ms)}$ onto $\mathcal{P}_\varepsilonprime$
\item $\mathcal{P}^{(\Ms+1)} \gets \mathcal{P}^{(\Ms)} \cup \{(\Cvq^\star,\CzSxR^{(\Ms)})\} $ \label{item:cpmod_newset}
\item Repeat \ref{item:cp_first})--\ref{item:cpmod_newset}) until
$\cpDelta^{(\Ms)}\leq\cpeps$
or $\mathcal{P}^{(\Ms+1)}=\mathcal{P}^{(\Ms)}$
\end{enumerate}
\end{algorithm}

When using the modified procedure in Algorithm~\ref{algo:cpmod},
problem \eqref{eq:cp2} with $\varepsilon=0$ directly provides an outer approximation of the original PDF rate maximization \eqref{eq:problem}.
This in contrast to Algorithm~\ref{algo:cp},
where \eqref{eq:cp2} provides an outer approximation of \eqref{eq:decomposition} which in turn is an inner approximation of \eqref{eq:problem}.
Therefore, no rigorous solution about the quality of the obtained solution is possible in case of Algorithm~\ref{algo:cp},
but the situation changes when using Algorithm~\ref{algo:cpmod}.

In this case, $\cpU^{(\Ms)}$ is an upper bound to the solution of \eqref{eq:problem}
and $\cpL^{(\Ms)}$ is a feasible point for \eqref{eq:problem}.
Thus, if $\cpDelta^{(\Ms)}\to0$, we have a guarantee that we have found a globally optimal solution to \eqref{eq:problem}.
However, the modification in Algorithm~\ref{algo:cpmod} comes at the cost that the convergence proof from \cite{Ke60} no longer applies.
If we have not yet found the optimum of the original problem, 
it is clear that an additional linear constraint due to a linearization at $(\Cvq^{(\Ms)},\CzSxR^{(\Ms)})$
refines the outer approximation (see Section~\ref{sec:algo:cp}),
but it is not clear whether a new linearization at $(\Cvq^\star,\CzSxR^{(\Ms)})$ does the same.
In fact, it might happen that the projections leads to a point $(\Cvq^\star,\CzSxR^{(\Ms)})$
that had already been contained in $\mathcal{P}^{(\Ms)}$ without being an optimal solution.
In this case, $\mathcal{P}^{(\Ms+1)}=\mathcal{P}^{(\Ms)}$, i.e., we would solve \eqref{eq:cp2} for the same outer approximation again in the next iteration
and would not get any further progress for $\cpU^{(\Ms+1)}$.

For this reason, the additional termination criterion $\mathcal{P}^{(\Ms+1)}=\mathcal{P}^{(\Ms)}$ was added to Algorithm~\ref{algo:cpmod}.
If the algorithm terminates due to this criterion, we have not managed to obtain a globally optimal solution to \eqref{eq:problem},
but the current $\cpDelta^{(\Ms)}$ gives a rigorous information about how far the distance to the globally optimal solution can be at most.
On the other hand, if the algorithm terminates due to the first criterion $\cpDelta^{(\Ms)}\leq\cpeps$,
we have a guarantee that the obtained solution is a globally optimal solution to the original PDF rate maximization \eqref{eq:problem}
up to the desired error tolerance $\cpeps$.
In particular, the latter will happen in cases where the optimal solution of \eqref{eq:problem} lies in $\mathcal{P}_\varepsilonprime$,
i.e., in cases where the optimal innovation covariance matrix has full rank.

\section{Results and Conclusion}
\label{sec:results}
For the numerical simulations, we adopt the line network example from \cite{WeGeUt13}, where the relay lies on a line between source and destination.
The distances source-relay, relay-destination, and source-destination are given by $\dRS = d, d \in (0,1)$, $\dDR = 1-d$, and $\dDS=1$, respectively. The channel matrices are given by $\mb{H}_\mt{AB} = d_\mt{AB}^{-\gamma/2} \mb{\tilde H}_\mt{AB}$ with $\gamma=4$
and $\mt{A},\mt{B}\in\{\nodeS,\nodeR,\nodeD\}$. The individual elements of each $\mb{\tilde H}_\mt{AB}$ are independent and circularly symmetric complex Gaussian distributed with zero mean and unit variance.

The first important observation of the simulations we performed is that Algorithm~\ref{algo:cpmod} always terminated with $\cpDelta^{(\Ms)}\leq\cpeps$
for $\cpeps=10^{-3}$.
This means that the proposed method manages to find the globally optimal PDF rate in the considered scenario,
i.e., the lack of a theoretical convergence guarantee does not hurt in this scenario.
Moreover, the difference between the PDF rates obtained with Algorithm~\ref{algo:cp} and Algorithm~\ref{algo:cpmod}
has always been much smaller than $\cpeps$ in the simulations,
i.e., 
Algorithm~\ref{algo:cp} manages to find the global optimum as well,
even though it has only been designed for computing a lower bound to the optimum.
As the numerical results do not differ, we do not present separate plots for the two algorithms below.

We can use the obtained globally optimal solutions to investigate two questions.
First, we can evaluate the performance of previously proposed suboptimal algorithms,
and second, we can study the gap between the PDF rate and the CSB.
As a suboptimal method, we have chosen the inner approximation approach (IAA) from \cite{WeGeUt13},
which was reported to outperform other suboptimal heuristics (see the simulations in \cite{WeGeUt13}).

\begin{figure}[!t]
	\centering
	\newlength\figureheight 
	\newlength\figurewidth 
	\setlength\figureheight{3.7cm} 
	\setlength\figurewidth{7.6cm}
	\definecolor{mycolor1}{rgb}{0.00000,0.44700,0.74100}%
\definecolor{gray}{RGB}{150,150,150}
\definecolor{gridgray}{RGB}{220,220,220}
\definecolor{tublue1}{RGB}{0, 153, 255}
\definecolor{tublue2}{RGB}{65, 190, 255}
\definecolor{tugreen1}{RGB}{145,172,107}
\definecolor{tured1}{RGB}{202,033,063}
\definecolor{tured2}{RGB}{229,052,024}
\definecolor{tuyellow}{RGB}{255,180,000}
\definecolor{tuorange}{RGB}{255,128,000}
\definecolor{tugreen2}{RGB}{181,202,130}
\begin{tikzpicture}

\begin{axis}[%
width=0.951\figurewidth,
height=\figureheight,
at={(0\figurewidth,0\figureheight)},
scale only axis,
xmin=-0.005,
xmax=0.099,
scaled ticks=false,
tick label style={/pgf/number format/fixed},
xlabel={$R_\text{proposed} - R_\text{IAA}$ [bpcu]},
ymin=0,
ymax=79,
ylabel={\hspace*{-0.4cm}absolute freq. of occurrences},
axis background/.style={fill=white},
xtick = {0,0.02,...,0.1},
minor xtick = {0.01,0.03,...,0.11},
grid style={thin,gridgray},
grid=both
]
\addplot[fill=gray,fill opacity=0.6,draw=black,ybar interval,area legend] plot table[row sep=crcr] {%
x	y\\
-0.0024	0\\
-0.001591	0\\
-0.000782	71\\
2.69999999999997e-05	23\\
0.000836	2\\
0.001645	0\\
0.002454	0\\
0.003263	0\\
0.004072	0\\
0.004881	0\\
0.00569	0\\
0.006499	0\\
0.007308	0\\
0.008117	0\\
0.008926	0\\
0.009735	0\\
0.010544	0\\
0.011353	0\\
0.012162	0\\
0.012971	0\\
0.01378	2\\
0.014589	1\\
0.015398	1\\
0.016207	0\\
0.017016	1\\
0.017825	2\\
0.018634	1\\
0.019443	0\\
0.020252	0\\
0.021061	0\\
0.02187	0\\
0.022679	1\\
0.023488	0\\
0.024297	0\\
0.025106	0\\
0.025915	0\\
0.026724	0\\
0.027533	0\\
0.028342	0\\
0.029151	0\\
0.02996	0\\
0.030769	0\\
0.031578	0\\
0.032387	0\\
0.033196	0\\
0.034005	0\\
0.034814	3\\
0.035623	0\\
0.036432	0\\
0.037241	0\\
0.03805	2\\
0.038859	0\\
0.039668	0\\
0.040477	1\\
0.041286	0\\
0.042095	1\\
0.042904	3\\
0.043713	3\\
0.044522	0\\
0.045331	0\\
0.04614	4\\
0.046949	1\\
0.047758	0\\
0.048567	2\\
0.049376	2\\
0.050185	0\\
0.050994	0\\
0.051803	1\\
0.052612	1\\
0.053421	1\\
0.05423	2\\
0.055039	2\\
0.055848	1\\
0.056657	2\\
0.057466	3\\
0.058275	2\\
0.059084	1\\
0.059893	2\\
0.060702	4\\
0.061511	4\\
0.06232	2\\
0.063129	0\\
0.063938	2\\
0.064747	4\\
0.065556	2\\
0.066365	2\\
0.067174	2\\
0.067983	2\\
0.068792	4\\
0.069601	1\\
0.07041	2\\
0.071219	5\\
0.072028	2\\
0.072837	4\\
0.073646	5\\
0.074455	1\\
0.075264	3\\
0.076073	0\\
0.076882	1\\
0.077691	2\\
0.0785	2\\
0.079309	0\\
0.089		1\\
0.089809		1\\
};
\end{axis}
\end{tikzpicture}%
	\caption{Histogram of rate gain over IAA \cite{WeGeUt13} for $\NS = \NR = \ND = 2$, $\PS = 100$, $\PR = 10$, $d=0.8$ and $\varepsilon = 10^{-5} \PS$.}
	\label{fig:hist}
\end{figure}
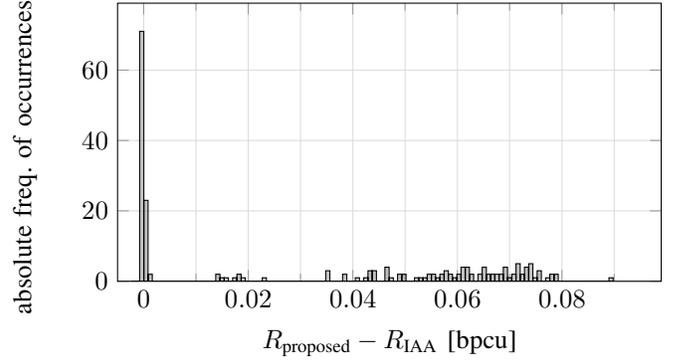
\begin{figure}[!t]
	\centering
	\setlength\figureheight{4.3cm} 
	\setlength\figurewidth{7.6cm}
	\definecolor{mycolor1}{rgb}{0.00000,1.00000,1.00000}%
\definecolor{tublue1}{RGB}{0, 153, 255}
\definecolor{tublue2}{RGB}{65, 190, 255}
\definecolor{tugreen1}{RGB}{145,172,107}
\definecolor{tured1}{RGB}{202,033,063}
\definecolor{tured2}{RGB}{229,052,024}
\definecolor{gridgray}{RGB}{220,220,220}
\definecolor{tuyellow}{RGB}{255,180,000}
\definecolor{tuorange}{RGB}{255,128,000}
\definecolor{tugreen2}{RGB}{181,202,130}
\begin{tikzpicture}

\begin{axis}[%
width=0.951\figurewidth,
height=\figureheight,
at={(0\figurewidth,0\figureheight)},
scale only axis,
xmin=0.1,
xmax=0.89,
xlabel={d},
ymin=14.2,
ymax=17.8,
ylabel={Rate [bpcu]},
axis background/.style={fill=white},
title style={font=\bfseries},
legend style={at={(0.01,0.99)}, anchor=north west,legend cell align=left,align=left,draw=white!15!black},
xtick = {0.2,0.4,...,0.8},
minor xtick = {0.3,0.5,0.7},
grid style={thin,gridgray},
grid=both
]
\addplot [color=black,thick, dotted]
  table[row sep=crcr]{%
0.1	14.2388403270632\\
0.2	14.7187837028188\\
0.3	15.3110129031721\\
0.4	16.0407105817475\\
0.5	16.8841226786877\\
0.6	17.4585061541446\\
0.7	16.9259419097039\\
0.8	15.9473861845729\\
0.9	15.1613946862957\\
};
\addlegendentry{cut-set bound};

\addplot [color=black,solid]
  table[row sep=crcr]{%
0.1	14.2276807237289\\
0.2	14.7176570462064\\
0.3	15.3083688476728\\
0.4	16.0251423838695\\
0.5	16.8016047616144\\
0.6	17.0835140824112\\
0.7	16.2822026044328\\
0.8	15.2793519036094\\
0.9	14.4913517381071\\
};
\addlegendentry{IAA DF};

\addplot [color=black,thick, dashed]
  table[row sep=crcr]{%
0.1	14.2422870850173\\
0.2	14.7179320365145\\
0.3	15.3081528497641\\
0.4	16.025750773699\\
0.5	16.8070472616593\\
0.6	17.1106658678561\\
0.7	16.3189145588717\\
0.8	15.3166627134353\\
0.9	14.5256311718428\\
};
\addlegendentry{proposed};

\end{axis}
\end{tikzpicture}%
	\caption{Average rate compared to IAA \cite{WeGeUt13} and to the CSB for $\NS = \NR = \ND = 2$, $\PS = 100$, $\PR = 10$, and $\varepsilon = 10^{-5} \PS$.\vspace{-1.5mm}}
	\label{fig:comparison}
\end{figure}
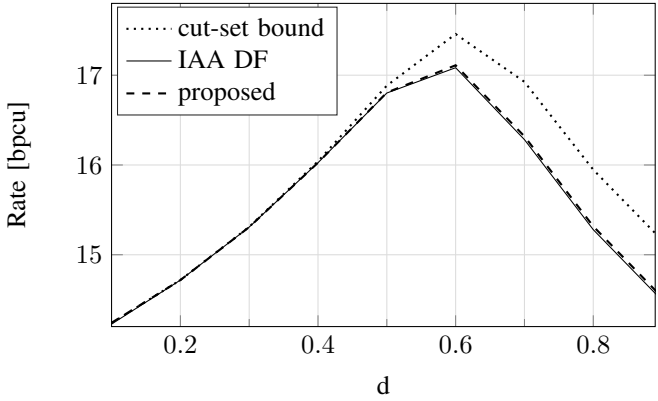

The histogram in Fig.~\ref{fig:hist} shows the difference $R_\mt{proposed}-R_\mt{IAA}$ for 200 i.i.d.\ channel realizations 
with two antennas at each terminal and distance parameter $d = 0.8$.
It can be seen that the IAA and the proposed algorithm converge to the same value in many cases.
However, there are also cases in which the proposed algorithm achieves a higher rate,
meaning that the solution found by the IAA method is not the global optimum in these cases.
Fig.~\ref{fig:comparison} shows the results for the same scenario with various values of $d$.
By using the proposed method as a benchmark, we can
conclude that the IAA has a close-to-optimal performance on average,
which had not been clear in the first place since the IAA is only a local method.
However, the IAA cannot find the global optimal for all channel realizations.

Concerning the gap to the CSB,
we can observe that this gaps remains for $d$ larger than approximately $0.5$ (see Fig.~\ref{fig:comparison})
even though we have solved the PDF rate maximization in a globally optimal manner.
This shows that this gap is not due to a potentially suboptimal choice of the covariance matrices, but it is either
inherent to the PDF scheme or inherent to the fact that the CSB might not be a tight bound to the capacity of the relay channel in general.

\section{Conclusion}
We have proposed a method to compute a solution to the problem of maximizing the partial decode-and-forward (PDF) rate
in the Gaussian MIMO relay channel.
In addition to the computed covariance matrices, the algorithm outputs an accuracy $\cpDelta^{(\Ms)}$
and guarantees that the achieved PDF rate is at most $\cpDelta^{(\Ms)}$ away from the true global optimum.
Even though the algorithm might theoretically terminate with $\cpDelta^{(\Ms)}\gg0$,
we have observed convergence $\cpDelta^{(\Ms)}\to0$ for all considered channel realizations of the considered numerical example.
This means, that the method indeed finds the globally optimal PDF rate in these scenarios.

An open topic for future research is to either find a formal convergence proof of the proposed method
or to derive an alternative solution approach for which a theoretical convergence guarantee can be given.
Another aspect is that it might be possible to calculate a Hessian matrix in addition to the gradient
and to use this second-order information to accelerate the numerical solution.

\appendices

\section{Derivation of \eqref{eq:detLambda}--\eqref{eq:detid}}
\label{app:detLambda}
Without loss of generality, assume that $\gevL$ in \eqref{eq:GEVDLambda} is arranged such that 
$\gevL=\blockdiag\{\gevL_1,\gevL_2\}$, where $\gevL_1$ contains all diagonal elements that are larger than one.
We can then write $\gev=\begin{bmatrix}\gev_1 & \gev_2\end{bmatrix}$, and we have
\begin{equation}
\gev_1\gev_1^+ = \gev \gevD \gev^\He \quad\text{with}\quad \gevD = \begin{bmatrix} (\gev_1^\He \gev_1)^{-1} & \zero \\ \zero & \zero \end{bmatrix}.
\end{equation}
Moreover, we can calculate
\begin{equation}
\gevD \gev^\He \gev  = \begin{bmatrix}
\id & (\gev_1^\He \gev_1)^{-1} \gev_1^\He \gev_2 \\ \zero & \zero
\end{bmatrix}
\end{equation}
and we obtain
\begin{align}
&\det (\id + \Gk \starCvS)\\&= \det(\id + \CvqS^\frac{1}{2} \gev_1 \gev_1^+ \CvqS^\frac{1}{2}
\CvqS^{-\frac{1}{2}} (\gev^{-\He} \gevL \gev^{-1} - \id) \CvqS^{-\frac{1}{2}}) \label{eq:det_proof_Lambda1}
\\&= \det(\id + \gev\gevD\gev^\He   (\gev^{-\He} \gevL \gev^{-1} - \id) )
\\&= \det(\id + \gevD\gev^\He   (\gev^{-\He} \gevL \gev^{-1} - \id)\gev )
\\&= \det(\id + \gevD \gevL - \gevD\gev^\He \gev ) \label{eq:det_proof_Lambda2}
\\&=\det\left(\begin{bmatrix} (\gev_1^\He \gev_1)^{-1} \gevL_1   & -(\gev_1^\He \gev_1)^{-1} \gev_1^\He \gev_2 \\ \zero & \id \end{bmatrix}\right)  \label{eq:det_proof_Lambda3}
\\&=\det((\gev_1^\He \gev_1)^{-1})\det(\gevL_1)  = \frac{\det(\gevLbar)}{\det(\gev_1^\He \gev_1)}.  \label{eq:det_proof_Lambda4}
\end{align}
To obtain the identity involving $\Gj$, we only need to replace $\gevL$ in \eqref{eq:det_proof_Lambda1}--\eqref{eq:det_proof_Lambda2} by $\id$,
yielding $\id$ instead of $\gevL_1$ in \eqref{eq:det_proof_Lambda3}--\eqref{eq:det_proof_Lambda4}.

\section{Proof of Proposition~\ref{prop:proj}}
\label{app:proj}
We solve the convex program \eqref{eq:proj} via its KKT conditions (e.g., \cite[Sec.~4.2]{BaShSh06}).
Introducing a Lagrangian multiplier matrix $\Lam\succeq\zero$,
the Lagrangian function of problem \eqref{eq:proj} read as
\begin{equation}
\Phi = \tr[ (\Cvq - \Cvq^{(\Ms)})^\He(\Cvq - \Cvq^{(\Ms)})] - \Lam (\Cvq-\varepsilonprime \idNS).
\end{equation}
The complementary slackness condition ${\Lam(\Cvq-\varepsilonprime \idNS)}=\zero$ implies
that $\Lam$ has the same eigenbasis as $(\Cvq-\varepsilonprime \idNS)$ and thus the same eigenbasis as $\Cvq$ (e.g., \cite[Proof of Th.~1]{HuScJoUt08}).
We thus can write $\Cvq=\mb W\diag\{\zeta_i\}\mb W^\He$ and $\Lam=\mb W\diag\{\lam_i\}\mb W^\He$.
Setting $\frac{\partial \Phi}{\partial \Cvq^\Tr}=\zero$, we have
\begin{mueq}
\Cvq - \Lam = \Cvq^{(\Ms)} \tclf ~\Leftrightarrow~ \mb W (\diag\{\zeta_i-\lam_i\}) \mb W^\He = \mb V \diag\{\kappa_i\} \mb V^\He
\end{mueq}
which implies that $\mb W=\mb V$ and $\zeta_i=\kappa_i+\lam_i$ with $\lam_i\geq0$.
For all $i$ with $\kappa_i<\varepsilonprime$, primal feasibility $\zeta_i\geq\varepsilonprime$ requires $\lam_i>0$,
and complementary slackness $\lam_i(\zeta_i-\varepsilonprime)=0$ then yields $\zeta_i=\varepsilonprime$.
For all $i$ with $\kappa_i\geq\varepsilonprime$, we have $\zeta_i-\varepsilonprime=\kappa_i+\lam_i-\varepsilonprime$,
so that complementary slackness yields $\lam_i=0$ and $\zeta_i=\kappa_i$.

\bibliographystyle{IEEEtran}
\bibliography{IEEEabrv,ConfIEEE,literature}

\end{document}